\documentclass[a4paper,USenglish]{lipics-v2019}
\title{Lower Bounds for Semi-adaptive Data Structures via Corruption}

\author{Pavel Dvo\v{r}\'{a}k}{Charles University, Prague, Czech Republic}{koblich@iuuk.mff.cuni.cz}{}{Supported by Czech Science Foundation 
GA{\v C}R (grant \#19-27871X).}% mandatory, please use full name; only 1 author per \author macro; first two parameters are mandatory, other parameters can be empty. Please provide at least the name of the affiliation and the country. The full address is optional

\author{Bruno Loff}{INESC-Tec and University of Porto, Portugal}{bruno.loff@gmail.com}{}{This project was financed by the Portuguese funding agency, FCT - Fundação para a Ciência e a Tecnologia, through national funds, and co-funded by the FEDER, where applicable. Bruno Loff was the recipient of the postdoctoral FCT fellowship number SFRH/BPD/116010/2016.}

\authorrunning{P. Dvo\v{r}\'{a}k and B. Loff} % mandatory. First: Use abbreviated first/middle names. Second (only in severe cases): Use first author plus 'et al.'

\Copyright{Pavel Dvo\v{r}\'{a}k and Bruno Loff} % mandatory, please use full first names. LIPIcs license is "CC-BY";  http://creativecommons.org/licenses/by/3.0/

\ccsdesc[100]{Theory of computation ~ Computational complexity and cryptography} % mandatory: Please choose ACM 2012 classifications from https://dl.acm.org/ccs/ccs_flat.cfm 

\keywords{semi-adaptive dynamic data structure, polynomial lower bound, corruption bound, information theory} % mandatory; please add comma-separated list of keywords

\category{} %optional, e.g. invited paper

\relatedversion{} %optional, e.g. full version hosted on arXiv, HAL, or other respository/website
\supplement{}%optional, e.g. related research data, source code, ... hosted on a repository like zenodo, figshare, GitHub, ...

\acknowledgements{We would like to thank Michal Koucký, who worked with us on this paper until the coronavirus pandemic forced him busily away. We would also like to thank Arkadev Chattopadhyay for helpful pointers.}

\nolinenumbers %uncomment to disable line numbering

\EventEditors{John Q. Open and Joan R. Access}
\EventNoEds{2}
\EventLongTitle{42nd Conference on Very Important Topics (CVIT 2016)}
\EventShortTitle{CVIT 2016}
\EventAcronym{CVIT}
\EventYear{2016}
\EventDate{December 24--27, 2016}
\EventLocation{Little Whinging, United Kingdom}
\EventLogo{}
\SeriesVolume{42}
\ArticleNo{23}
\newtheorem{fact}[theorem]{Fact}
\newtheorem{observation}[theorem]{Observation}

\newcommand{\R}{\mathbb{R}}

\DeclareMathOperator*{\expec}{\mathbb{E}}

\newcommand{\Disj}{\mathsf{DISJ}}
\newcommand{\IP}{\mathsf{IP}}
\newcommand{\GH}{\mathsf{GHD}}
\newcommand{\ORT}{\mathsf{ORT}}
\newcommand{\mono}{\mathsf{mono}}
\newcommand{\disc}{\mathsf{disc}}

\newcommand{\cb}{\mathsf{cb}}
\newcommand{\scb}{\mathsf{scb}}

\newcommand{\tq}{t_{\mathrm q}}
\newcommand{\tu}{t_{\mathrm u}}
\newcommand{\tp}{t_{\mathrm p}}

\newcommand{\bfx}{{\mathbf x}}
\newcommand{\bfy}{{\mathbf y}}
\newcommand{\bfz}{{\mathbf z}}
\newcommand{\rmx}{{\mathrm x}}
\newcommand{\rmy}{{\mathrm y}}
\newcommand{\rmz}{{\mathrm z}}

\newcommand{\supp}{{\mathsf{supp}}}

\newcommand{\DKL}{{\mathsf{D}_{\mathsf{KL}}}}

\newcommand*\bfell{\ensuremath{\boldsymbol\ell}}

\newtheorem*{observation*}{Observation}
\newtheorem*{lemma*}{Lemma}

\begin{document}
\maketitle
\begin{abstract}
 In a dynamic data structure problem we wish to maintain an encoding of some data in memory, in such a way that we may efficiently carry out a sequence of queries and updates to the data.
A long-standing open problem in this area is to prove an unconditional polynomial lower bound of a trade-off between the update time and the query time of an adaptive dynamic data structure computing some explicit function.
Ko and Weinstein provided such lower bound for a restricted class of {\em semi-adaptive\/} data structures, which compute the Disjointness function. 
There, the data are subsets $x_1,\dots,x_k$ and $y$ of $\{1,\dots,n\}$, the updates can modify $y$ (by inserting and removing elements), and the queries are an index $i \in \{1,\dots,k\}$ (query $i$ should answer whether $x_i$ and $y$ are disjoint, i.e., it should compute the Disjointness function applied to $(x_i, y)$). The semi-adaptiveness places a restriction in how the data structure can be accessed in order to answer a query. We generalize the lower bound of Ko and Weinstein to work not just for the Disjointness, but for any function having high complexity under the smooth corruption bound.

%%% Local Variables:
%%% mode: latex
%%% TeX-master: "paper"
%%% End:

\end{abstract}

\section{Introduction}

In a dynamic data structure problem we wish to maintain an encoding of some data in memory, in such a way that we may efficiently carry out a sequence of queries and updates to the data.
 A suitable computational model to study dynamic data structures is the cell probe model of Yao \cite{Yao79}. Here we think of the memory divided into registers, or \emph{cells}, where each cell can carry $w$ bits, and we measure efficiency by counting the number of memory accesses, or \emph{probes}, needed for each query and each update --- these are respectively called the \emph{query time} $\tq$ and \emph{update time} $\tu$.
The main goal of this line of research is to understand the inherent trade-off between $w$, $\tq$ and $\tu$, for various interesting problems. Specifically, one would like to show lower bounds on $t = \max\{ \tq, \tu \}$ for reasonable choices of $w$ (which is typically logarithmic in the size of the data).

The first lower bound for this setting was proven by Fredman and Saks \cite{fredman1989cell}, which proved $t = \Omega\bigl(\log n / \log \log n\bigr)$ for various problems. These lower bounds were successively improved \cite{puaatracscu2004tight,patrascu2006logarithmic,larsen2012cell,larsen2020crossing}, and we are now able to show that certain problems with non-Boolean queries require $t = \Omega\bigl((\log n/\log \log n)^2\bigr)$, and certain problems with Boolean queries require $t = \Omega\bigl((\log n/\log \log n)^{3/2}\bigr)$.

The major unsolved question in this area is to prove a polynomial lower bound on $t$. For example, consider the dynamic reachability problem, where we wish to maintain a directed $n$-vertex graph in memory, under edge insertions and deletions, while being able to answer reachability queries (``\emph{is vertex $i$ connected to vertex $j$?''}). Is it true that any scheme for the dynamic reachability problem requires $t = \Omega(n^{\delta})$, for some constant $\delta > 0$? Indeed, such a lower bound is known under various complexity-theoretic assumptions\footnote{See \cite{patrascu2010towards,abboud2014popular}. Strictly speaking, these conditional lower bounds only work if the preprocessing time, which is the time taken to encode the data into memory, is also bounded. But we will ignore this distinction.}, the question is whether such a lower bound may be proven unconditionally.

In an influential paper~\cite{patrascu10}, Mihai P\u{a}tra\c{s}cu proposed an approach to this unsolved question. He defines a data structure problem, called the \emph{multiphase problem}. %Let us define it now.
Let us represent partial functions $f: \{0,1\}^n \times \{0,1\}^n \to \{0,1\}$ as total functions $f': \{0,1\}^n \times \{0,1\}^n \to \{0,1,*\}$ where $f'(x,y) = *$ if $f(x,y)$ is not defined.
Then associated with a partial Boolean function $f: \{0,1\}^n \times \{0,1\}^n \to \{0,1, \ast\}$, and a natural number $k \ge 1$, we may define a corresponding \emph{multiphase problem of $f$} as the following dynamic process:
\begin{description}
\item[Phase I - Initialization.] We are given $k$ inputs $x_1,\dots,x_k \in \{0,1\}^n$, and are allowed to preprocess this input in time $n k \cdot \tp$.
\item[Phase II - Update.] We are then given another input $y \in \{0,1\}^n$, and we have time $n\cdot \tu$ to read and update the memory locations from the data structure constructed in Phase I.
\item[Phase III - Query.] Finally, we are given a query $i \in [k]$, we have time $\tq$ to answer the question whether $f(x_i,y) = 1$. If $f(x_i,y)$ is not defined, the answer can be arbitrary.
\end{description}
Typically we will have $k = \text{poly}(n)$. Let us be more precise, and consider randomized solutions to the above problem.

\begin{definition}[Scheme for the multiphase problem of $f$] Let $f: \{0,1\}^n \times \{0,1\}^n \to \{0,1, \ast\}$ be a partial Boolean function.
  A \emph{scheme} for the multiphase problem of $f$ with preprocessing time $\tp$, update time $\tu$ and query time $\tq$ is a triple $D = \bigl(E, \{U_y\}_{y \in \{0,1\}^n}, \{Q_i\}_{i \in [k]}\bigr)$, where:
  \begin{itemize}
  \item $E: \bigl(\{0,1\}^n\bigr)^k \to \bigl(\{0,1\}^w\bigr)^s$ maps the input $x$ to the memory contents $E(x)$, where each of the $s$ memory locations holds $w$ bits. $E$ must be computed in time $n k \cdot \tp$ by a Random-Access Machine (RAM).
  \item For each $y \in \{0,1\}^n$, $U_y:\bigl(\{0,1\}^w\bigr)^s \to \bigl(\{0,1\}^w\bigr)^u$ is a decision-tree of depth $\le n \cdot \tu$, which reads $E(x)$ and produces a sequence $U_y\bigl(E(x)\bigr)$ of $u$ \emph{updates}.\footnote{In the usual way of defining the update phase, we have a read/write decision-tree $U_y$ which changes the very same cells that it reads. But when $w = \Omega(\log s)$, this can be seen to be equivalent, up to constant factors, to the definition we present here, where we have a decision-tree $U_y$ that writes the updates on a separate location. In order to simulate a scheme that uses a read/write decision-tree, we may use a hash table with $O(1)$ worst-case lookup time, such as cuckoo hashing. Then we have a read-only decision-tree $U'_y(E(x))$ whose output is the hash table containing all the $i \in [s]$ which were updated by $U_y(E(x))$, associated with their final value in the execution of $U_y(E(x))$. Note that the hash table itself is static. }
  \item For each $i \in [k]$, $Q_i:\bigl(\{0,1\}^w\bigr)^s \times \bigl(\{0,1\}^w\bigr)^u \to \{0,1\}$ is a decision-tree of depth $\le \tq$.\footnote{All our results will hold even if $Q_i$ is allowed to depend arbitrarily on $x_i$. This makes for a less natural model, however, so we omit this from the definitions.}
  \item For all $x \in \bigl(\{0,1\}^n\bigr)^k$, $y \in \{0,1\}^n$, and $i \in [k]$,
    \[
      f(x_i, y) \neq * \implies Q_i\bigl(E(x), U_y(E(x))\bigr) = f(x_i, y).
    \]
  \end{itemize}

  In a \emph{randomized scheme} for the multiphase problem of $f$, each $U_y$ and $Q_i$ are distributions over decision trees, and it must hold that for all $x \in \bigl(\{0,1\}^n\bigr)^k$, $y \in \{0,1\}^n$, and $i \in [k]$,
    \[
      f(x_i, y) \neq * \implies \Pr_{Q_i, U_y}\bigl[ Q_i\bigl(E(x), U_y(E(x))\bigr) = f(x_i, y) \bigr] \ge 1 - \varepsilon.
    \]
    The value $\varepsilon$ is called the \emph{error probability} of the scheme.
\end{definition}
% }

\bigskip\noindent 
P\u{a}tra\c{s}cu~\cite{patrascu10} considered this problem where $f = \Disj$ is the Disjointness function:
\[
 \Disj(x,y) = 
 \begin{cases}
  0 & \text{if there exists $i \in [n]$ such that $x_i = y_i = 1$} \\
  1 & \text{otherwise}
 \end{cases}
\]
He conjectured that any scheme for the multiphase problem of $\Disj$ will necessarily have $\max\{\tp, \tu, \tq\} \geq n^\delta$ for some constant $\delta > 0$.

P\u{a}tra\c{s}cu shows that such lower bounds on the multiphase problem for $\Disj$ would imply polynomial lower bounds for various dynamic data structure problems. For example such lower bounds would imply that dynamic reachability requires $t = \Omega(n^\delta)$. He also shows that these lower bounds hold true under the assumption that 3SUM has no sub-quadratic algorithms.

Finally, P\u{a}tra\c{s}cu then defines a 3-player Number-On-Forehead (NOF) communication game, such that lower bounds on this game imply matching lower bounds for the multiphase problem. The game associated with a function $f:\{0,1\}^n\times\{0,1\}^n\to\{0,1\}$ is as follows:
\begin{enumerate}
 \item Alice is given $x_1,\dots,x_k \in \{0,1\}^n$ and $i \in [k]$, Bob gets $y\in \{0,1\}^n$ and $i \in [k]$ and Charlie gets $x_1,\dots,x_k$ and $y$.
 \item Charlie sends a private message of $\ell_1$ bits to Bob and then he is silent.
 \item Alice and Bob communicate $\ell_2$ bits and want to compute $f(x_i,y)$.
\end{enumerate}
P\u{a}tra\c{s}cu~\cite{patrascu10} conjectured that if $\ell_1$ is $o(k)$, then $\ell_2$ has to be bigger than the communication complexity of $f$.
However, this conjecture turned out to be false. The randomized communication complexity of $\Disj$ is $\Omega(n)$ \cite{razborov92,kalyanasundaram92,bar-yossef02}, but Chattopadhyay et al.~\cite{chattopadhyay12} construct a protocol for $f = \Disj$ where both $\ell_1, \ell_2 = O\bigl(\sqrt{n}\cdot \log k\bigr)$. They further show that any randomized scheme in the above model can be derandomized.

So the above communication model is more powerful than it appears at first glance.\footnote{The conjecture remains that if $\ell_1 = o(k)$, then $\ell_2$ has to be larger than the maximum distributional communication complexity of $f$ under a product distribution. This is $\tilde\Theta\bigl(\sqrt n\bigr)$ for Disjointness~\cite{babai86}.} 
However, a recent paper by Ko and Weinstein~\cite{ko19} succeeds in proving lower bounds for a simpler version of the multiphase problem, which translate to lower bounds for a restricted class of dynamic data structure schemes. They manage to prove a lower bound of $\Omega(\sqrt n)$ for the simpler version of the multiphase problem which is associated with the Disjointness function $f = \Disj$. Our paper generalizes their result:
\begin{itemize}
\item We generalize their lower bound to any function $f$ having large complexity according to the smooth corruption bound, under a product distribution. Disjointness is such a function \cite{babai86}, but so is Inner Product, Gap Orthogonality, and Gap Hamming Distance~\cite{sherstov11}.
\item The new lower-bounds we obtain (for Inner-product, Gap Orthogonality, and Gap Hamming Distance) are stronger --- $\Omega(n)$ instead of the lower-bound $\Omega(\sqrt n)$ for disjointness. As far as was known before our result, it could well have been that every function had a scheme for the simpler version of the multiphase problem using only $O(\sqrt n)$ communication.
\item Ko and Weinstein derive their lower-bound via a cut-and-paste lemma which works specifically for disjointness. This cut-and-paste lemma is a more robust version of the one appearing in \cite{bar-yossef02}, made to work not only for protocols, where the inputs $\bfx$ and $\bfy$ are independent given the transcript $\bfz$ of the protocol, but also for random-variables that are ``protocol-like'', namely any $(\bfx,\bfy,\bfz)$ where $I(\bfx:\bfy\mid \bfz)$ is close to $0$. Instead, we directly derive the existence of a large nearly-monochromatic rectangle, from the existence of such protocol-like random-variables, which is what then allows us to use the smooth corruption bound. This result is our core technical contribution, and may be of independent interest.
\end{itemize}

All of the above lower bounds will be shown to hold also for randomized schemes, and not just for deterministic schemes.

\subsection{Semi-adaptive Multiphase Problem}
Let us provide rigorous definitions.

\begin{definition}[Semi-adaptive random data structure~\cite{ko19}]
  Let $f: \{0,1\}^n \times \{0,1\}^n \to \{0,1, *\}$ be a partial function. A scheme $D = \bigl(E, \{U_y\}_{y \in \{0,1\}^n}, \{Q_i\}_{i \in [k]}\bigr)$ for the multiphase problem of $f$ is called \emph{semi-adaptive} if any path on the decision-tree $Q_i:\bigl(\{0,1\}^w\bigr)^s \times \bigl(\{0,1\}^w\bigr)^u \to \{0,1\}$ first queries the first part of the input (the $E(x)$ part), and then queries the second part of the input (the $U(E(x))$ part). If $D$ is randomized, then this property must hold for every randomized choice of $Q_i$.
\end{definition}
We point out that the reading of the cells in each part is completely adaptive. The restriction is only that the data structure cannot read cells of $E(x)$ if it already started to read cells of $U(E(x))$.
Ko and Weinstein state their result for deterministic data structures, i.e., $\varepsilon = 0$ thus the data structure always returns the correct answer.

\begin{theorem}[Theorem 4.9 of Ko and Weinstein~\cite{ko19}]
 \label{thm:lb_ds}
 Let $k \geq \omega(n)$.
 Any semi-adaptive deterministic data structure that solves the multiphase problem of the $\Disj$ function, must have either $\tu \cdot n \geq \Omega\bigl(k/w\bigr)$ or $\tq \geq \Omega\bigl(\sqrt{n}/w\bigr)$.
\end{theorem}
To prove the lower bound they reduce the semi-adaptive data structure into a low correlation random process.

\begin{theorem}[Reformulation of Lemma 4.1 of Ko and Weinstein~\cite{ko19}]
\label{thm:rand_var}
Let $\bfx_1,\dots, \bfx_k$ be random variables over $\{0,1\}^n$ and each of them is independently distributed according to the same distribution $\mu_1$ and let $\bfy$ be a random variable over $\{0,1\}^n$ distributed according to $\mu_2$ (independently of $\bfx_1,\dots, \bfx_k$).
Let $D$ be a randomized semi-adaptive scheme for the multiphase problem for a partial function $f: \{0,1\}^n \times \{0,1\}^n \to \{0,1,\ast\}$ with error probability bounded by $\varepsilon$.
Then, for any $p \leq o(k)$ there is a random variable $\bfz \in \{0,1\}^m$ and $i \in [k]$ such that:
\begin{enumerate}
 \item $\Pr\bigl[f(\bfx_i,\bfy) \neq *, ~\bfz_m \neq f(\bfx_i, \bfy)  \bigr] \leq \varepsilon$.
 \item $I\bigl(\bfx_i : \bfy\, \bfz \bigr) \leq \tq \cdot w + o(\tq \cdot w)$.
 \item $I\bigl(\bfy : \bfz\bigr) \leq \tq \cdot w$.
 \item $I\bigl(\bfx_i : \bfy \mid \bfz\bigr) \leq O\bigl(\frac{\tu \cdot n \cdot w}{p}\bigr)$.
\end{enumerate}
\end{theorem}

Above, $I$ stands for mutual information between random variables, see Section \ref{sec:information-theory} for the definition. 
The random variable $\bfz$ consists of some $\bfx_j$'s and transcripts of query phases of $D$ for some $j \in [k]$.
The theorem can be interpreted as saying that the last bit of $\bfz$ predicts $f(\bfx_i, \bfy)$, $\bfz$ has little information about $\bfx_i$ and $\bfy$, and the triple $(\bfx_i, \bfy, \bfz)$ is ``protocol-like'', in the sense that $\bfx_i$ and $\bfy$ are close to being independent given $\bfz$.
Ko and Weinstein~\cite{ko19} proved Theorem~\ref{thm:rand_var} for the deterministic schemes for the $\Disj$ function and in the case where $\mu_1 = \mu_2$. 
However, their proof actually works for any (partial) function $f$ and for any two, possibly distinct distributions $\mu_1$ and $\mu_2$.
Moreover, their proof also works for randomized schemes. The resulting statement for randomized schemes for any function $f$ is what we have given above. %The last bit of $\bfz$ is the answer of $D$ on the query $i$, and if $D$ is randomized then $\bfz_m$ equals to $f(\bfx_i, \bfy)$ with high probability ($1 - \varepsilon$).
To complete the proof of their lower bound, Ko and Weinstein proved that if we set $p$ (and $k$) large enough so that $I\bigl(\bfx_i : \bfy \mid \bfz\bigr) \leq o(1)$ then such random variable $\bf{z}$ cannot exist when $f$ is the $\Disj$ function.
It is this second step which we generalize.

Let $f: X \times Y \to \{0,1\}$ be a function and $\mu$ be a distribution over $X \times Y$.
A set $R \subseteq X \times Y$ is a \emph{rectangle} if there exist sets $A \subseteq X$ and $B \subseteq Y$ such that $R = A \times B$.
\label{b-monochromatic}For $b \in \{0,1\}$ and $0 \leq \rho \leq 1$, we say the rectangle $R$ is \emph{$\rho$-error $b$-monochromatic for $f$ under $\mu$} if $\mu\bigl(R \cap f^{-1}(1-b)\bigl) \leq \rho \cdot \mu\bigl(R\bigr)$.
We say the distribution $\mu$ is a \emph{product distribution} if there are two independent distribution $\mu_1$ over $X$ and $\mu_2$ over $Y$ such that $\mu(x,y) = \mu_1(x) \times \mu_2(y)$.
For $0 \leq \alpha \leq \frac{1}{2}$, the distribution $\mu$ is $\alpha$-balanced according to $f$ if $\mu\bigl(f^{-1}(0)\bigr), \mu\bigl(f^{-1}(1)\bigr) \geq \alpha$.
We will prove that the existence of a random variable $\bfz$ given by Theorem~\ref{thm:rand_var} implies that, for any $b \in \{0,1\}$, any balanced product distribution $\mu$ and any function $g$ which is ``close'' to $f$, there is a large (according to $\mu$) $\rho$-error $b$-monochromatic rectangle for $g$ in terms of $\tq$.
This technique is known as smooth corruption bound~\cite{beame06, chakrabarti12} or smooth rectangle bound~\cite{klauck10}.
We denote the smooth corruption bound of $f$ as $\scb^{\rho,\lambda}_\mu$.
Informally, $\scb^{\rho,\lambda}_\mu(f) \geq s$ if there is $b \in \{0,1\}$ and a partial function $g: X \times Y \to \{0,1, *\}$ which is close\footnote{``Closeness'' is measured by the parameter $\lambda \in \R$, see Section \ref{sec:lambda} for the formal definition.} to $f$ such that any $\rho$-error $b$-monochromatic rectangle $R \subseteq X \times Y$ for $g$ has size (under $\mu$) at most $2^{-s}$.
We will define smooth corruption bound formally  in the next section.
Thus, if we use Theorem~\ref{thm:rand_var} as a black box we generalize Theorem~\ref{thm:lb_ds} for any function of large corruption bound.

% \begin{definition}
%  Let $f: X \times Y \to \{0,1\}$ be a function, $R = A \times B \subseteq X \times Y$ be a rectangle and $\mu$ be a distribution over $X \times Y$.
%  The rectangle $R$ is \emph{$\varepsilon$-almost $b$-monochromatic} for $f$ under $\mu$ if $\mu\bigl(R \setminus f^{-1}(b)\bigr) \leq \varepsilon \cdot \mu(R)$.
% \end{definition}
% 
% \begin{definition}
%  For a function $f: X \times Y \to \{0,1\}$ and a distribution $\mu$ over $X \times Y$ we define 
%  \[
%   \textit{mono}^\varepsilon_\mu(f) = \min_{b \in \{0,1\}} \max \bigl\{\mu(R) | R \text{ is an $\varepsilon$-almost $b$-monochromatic rectangle for $f$ under $\mu$}\bigr\}.
%  \]
% \end{definition}
% 
% \begin{definition}
% Let $\alpha > 0$.
% We say that a distribution $\mu$ over $S$ is $\alpha$-balanced according to a function $f : S \to \{0,1\}$ if $\mu\bigl(f^{-1}(0)\bigr), \mu\bigl(f^{-1}(1)\bigr) \geq \alpha$.
% \end{definition}
% 
% \begin{definition}
%  We say that a distribution $\mu$ over $X \times Y$ is a product distribution if there are two independent distribution $\mu_1$ over $X$ and $\mu_2$ over $Y$ such that $\mu = \mu_1 \times \mu_2$.
% \end{definition}

\begin{theorem}[Main Result]
\label{thm:MainResult}
 Let $\lambda, \tilde{\varepsilon},\tilde{\alpha} \geq 0$ such that $\alpha \geq 2\varepsilon$ for $\varepsilon = \tilde\varepsilon + \lambda, \alpha = \tilde\alpha - \lambda$.
 Let $\mu$ be a product distribution over $\{0,1\}^n \times \{0,1\}^n$ such that $\mu$ is $\tilde\alpha$-balanced according to a partial function $f: \{0,1\}^n \times \{0,1\}^n \to \{0,1, *\}$.
 Any semi-adaptive randomized scheme for the multiphase problem of $f$, with error probability bounded by $\tilde\varepsilon$, must have either $\tu \cdot n \geq \Omega\bigl(k/w\bigr)$, or 
 \[
  \tq\cdot w \geq \Omega \left( \alpha \cdot \scb^{O(\varepsilon/\alpha),\lambda}_\mu(f) \right).
 \]
\end{theorem}
We point out that $\Omega$ and $O$ in the bound given above hide absolute constants independent of $\alpha, \varepsilon$ and $\lambda$.

As a consequence of our main result, and of previously-known bounds on corruption, we are able to show new lower-bounds of $\tq = \Omega(\frac{n}{w})$ against semi-adaptive schemes for the multiphase problem of the Inner Product, Gap Orthogonality and Gap Hamming Distance functions (where the gap is $\sqrt{n}$). These lower-bounds hold assuming that $\tu = o(\frac{k}{w n})$. They follow from the small discrepancy of Inner Product, and from a bound shown by Sherstov on the corruption of Gap Orthogonality, followed by a reduction to Gap Hamming Distance \cite{sherstov11}. This result also gives an alternative proof of the same lower-bound proven by Ko and Weinstein \cite{ko19}, for the Disjointness function, of $\tq = \Omega(\frac{\sqrt n}{w})$. This follows from the bound on corruption of Disjointness under a product distribution, shown by Babai et al.~\cite{babai86}.

The paper is organized as follows. In Section \ref{sec:preliminaries} we give important notation, and the basic definitions from information theory and communication complexity. The proof of Theorem \ref{thm:MainResult} appears in Section \ref{sec:proof}. The various applications appear in Section \ref{sec:App}.

\section{Preliminaries}\label{sec:preliminaries}

We use a notational scheme where sets are denoted by uppercase letters, such as $X$ and $Y$, elements of the sets are denoted by the same lowercase letters, such as $x \in X$ and $y \in Y$, and random variables are denoted by the same lowercase boldface letters, such as $\bfx$ and $\bfy$. We will use lowercase greek letters, such as $\mu$, to denote distributions. If $\mu$ is a distribution over a product set, such as $X\times Y\times Z$, and $(x,y,z) \in X\times Y\times Z$, then $\mu(x,y,z)$ is the probability of seeing $(x,y,z)$ under $\mu$. We will sometimes denote $\mu$ by $\mu(\rmx, \rmy, \rmz)$, using non-italicized lowercase letters corresponding to $X \times Y \times Z$. This allows us to to use the notation $\mu(\rmx)$ and $\mu(\rmy)$ to denote the $\rmx$ and $\rmy$-marginals of $\mu$, for example; then if we use the same notation with italicized lowercase letters, we get the marginal probabilities, i.e., for each $x \in X$ and $y \in Y$
\[
  \mu(x) = \sum_{y,z} \mu(x,y,z) \qquad \mu(y) = \sum_{x,z} \mu(x,y,z).
\]
If $y \in Y$, then we will also use the notation $\mu(\rmx \mid y)$ to denote the $\rmx$-marginal of $\mu$ conditioned seeing the specific value $y$. Then for each $x \in X$ and $y \in Y$, we have
\[
  \mu(x\mid y) = \sum_z \mu(x,y,z).
\]
We will also write $(\bfx,\bfy,\bfz) \sim \mu$ to mean that $(\bfx,\bfy,\bfz)$ are random variables chosen according to the distribution $\mu(\rmx,\rmy,\rmz)$, i.e., for all $(x,y,z)\in X\times Y\times Z$, $\Pr[\bfx = x, \bfy = y, \bfz = z] = \mu(x,y,z)$. Naturally if $A \subseteq X\times Y\times Z$, then $\mu(A) = \sum_{(x,y,z)\in A} \mu(x,y,z)$. We let $\supp(\mu)$ denote the \emph{support} of $\mu$, i.e., the set of $(x,y,z)$ with $\mu(x,y,z) > 0$.

We now formally define the smooth corruption bound and related measures from communication complexity, and refer the book by Kushilevitz and Nisan~\cite{kushilevitz96} for more details.
At the end of this section we provide necessary notions of information theory which are used in the paper, and for more details on these we refer to the book by Cover and Thomas~\cite{cover06}.

%A \emph{rectangle} $R$ is a product set, so $R = A\times B$ for some sets $A$ and $B$. If we write $R \subseteq X\times Y$ and say $R$ is a rectangle, we mean that $R=A\times B$ with $A \subseteq X$ and $B \subseteq Y$.

\subsection{Rectangle Measures}\label{sec:lambda}
Let $f: X \times Y \to \{0,1, *\}$ be a partial function, where $f(x,y) = \ast$ means $f$ is not defined on $(x,y)$. Let $\mu(\rmx,\rmy)$ be a distribution over $X \times Y$.
We say that $f$ is $\lambda$-close to a partial function $g: X \times Y \to \{0,1,*\}$ under $\mu$ if 
\[
\Pr_{(x,y) \sim \mu} \bigl[f(x,y) \neq g(x,y)\bigr] \leq \lambda. 
\]
For $b \in \{0,1\}$, $\rho \in [0,1]$, let
\[
  {\cal R}^{\rho,b}_\mu(f) = \bigl\{ R \subseteq X\times Y \text{ rectangle} \mid \mu\bigl(R \cap f^{-1}(1-b)\bigl) \leq \rho\cdot \mu\bigl(R\bigr)\bigr\}
\]
be the set of $\rho$-error $b$-monochromatic rectangles for $f$ under $\mu$.
The complexity measure $\mono$ quantifies how large almost $b$-monochromatic rectangles can be for both $b \in \{0,1\}$:
\[
 \mono^\rho_\mu(f) = \min_{b \in \{0,1\}} \max_{R \in {\cal R}^{\rho,b}_\mu(f)} \mu(R)
\]
Using $\mono$ we can define the \emph{corruption bound} of a function as $\cb^\rho_\mu(f) = \log \frac{1}{\mono^\rho_\mu(f)}$ and the \emph{smooth corruption bound} as
\[
\scb^{\rho,\lambda}_\mu(f) = \max_{g: \text{ $\lambda$-close to $f$ under $\mu$} } \cb^\rho_\mu(g).
\]
Thus, if $\scb^{\rho,\lambda}_\mu(f) \geq s$ then there is a $b \in \{0,1\}$ and a function $g$ which $\lambda$-close to $f$ under $\mu$ such that for any $\rho$-error $b$-monochromatic rectangle for $g$ under $\mu$ it holds that $\mu(R) \leq 2^{-s}$.

\begin{remark*}
  In Razborov's paper where an $\Omega(n)$ lower-bound for disjointness is first proven \cite{razborov92}, the (implicitly given) definition of a $\rho$-error $b$-monochromatic rectangle is $\mu(R\cap f^{-1}(1-b)) \le \rho \cdot \mu(R \cap f^{-1}(b))$. Later, a strong direct product theorem for corruption (under product distributions) was proven by Beame et al.~\cite{beame06}, which uses instead the condition that $\mu(R \setminus f^{-1}(b)) \le \rho \cdot \mu(R)$. The definition we present above comes from \cite{sherstov11}, where the condition is (we repeat it here) that $\mu(R \cap f^{-1}(1-b)) \leq \rho\cdot \mu(R)$. So we have three different definitions of $\rho$-error $b$-monochromatic rectangle, and thus three different corruption bounds. Now, if the distribution $\mu$ is supported on the domain of $f$, all these three definitions result in (roughly) equivalent complexity measures. But if $\mu$ attributes some mass to inputs where $f$ is undefined (which is sometimes useful if $\mu$ is a product distribution, as in our case), then the definitions are no longer equivalent. Our lower-bound will hold for any of the definitions, but the proof is somewhat simpler for the definition used in Sherstov's paper \cite{sherstov11}, which is the only corruption-based lower-bound we use, where $\mu$ attributes mass to undefined inputs.
\end{remark*}

The notion $\mono^\rho_\mu$ is related to the \emph{discrepancy} of a function:
\[
 \disc_\mu(f) = \max_{\text{$R:$ rectangle of $X \times Y$ }} \Bigl| \mu\bigl(R \cap f^{-1}(0)\bigr) - \mu\bigl(R \cap f^{-1}(1)\bigr)\Bigr|.
\]
It is easy to see that for a total function $f$ holds that $\disc_\mu(f) \geq (1-2\rho)\cdot \mono^\rho_\mu(f)$ for any $\rho$.
Thus, Theorem~\ref{thm:MainResult} will give us lower bounds also for functions of small discrepancy.

\subsection{Information Theory}\label{sec:information-theory}
We define several measures from information theory.
If $\mu'(\rmz), \mu(\rmz)$ are two distributions such that $\supp(\mu') \subseteq \supp(\mu)$, then the \emph{Kullback-Leibler divergence} of $\mu'$ from $\mu$ is
 \[
  \DKL \bigl(\mu'~ \|~ \mu\bigr) = \sum_{z} \mu'(z) \log \frac{\mu'(z)}{\mu(z)}.
 \]

 With Kullback-Leibler divergence we can define the mutual information, which measures how close (according to KL divergence) is a joint distribution to the product of its marginals.
 If we have two random variables $(\bfx,\bfy) \sim \mu(\rmx, \rmy)$, then we define their \emph{mutual information} to be
\[
  I\bigl({\bfx : \bfy}\bigr) = \DKL \bigl(\mu(\rmx, \rmy) ~\|~ \mu(\rmx) \times \mu(\rmy)\bigr) = \expec_{y \sim \mu(\rmy)} \Big[ \DKL\bigl(\mu(\rmx\mid y) ~\|~\mu(\rmx)\bigr) \Big].
\]
If we have three random variables $(\bfx, \bfy, \bfz) \sim \mu(\rmx, \rmy, \rmz)$, then the \emph{mutual information of $\bfx$ and $\bfy$ conditioned by $\bfz$} is
\[
  I\bigl(\bfx : \bfy \mid \bfz\bigr) = \expec_{z \sim \mu(\rmz)} \Big[ I\bigl(\bfx : \bfy \mid \bfz = z\bigr) \Big] = \expec_{z \sim \mu(\rmz)} \Big[ \DKL\Big( \mu(\rmx, \rmy \mid z) ~\|~ \mu(\rmx \mid z) \times \mu( \rmy \mid z)\Big) \Big]
\]
We present several facts about mutual information, the proofs can be found in the book of Cover and Thomas~\cite{cover06}.

\begin{fact}[Chain Rule]
 For any random variables $\bfx_1,\bfx_2,\bfy$ and $\bfz$ holds that
 \[
  I\bigl(\bfx_1 \bfx_2 : \bfy \mid \bfz\bigr) = I\bigl(\bfx_1 : \bfy \mid \bfz\bigr) + I\bigl(\bfx_2 : \bfy \mid \bfz, \bfx_1\bigr).
 \]
\end{fact}
\noindent Since mutual information is never negative, we have the following corollary.
\begin{corollary}
 For any random variables $\bfx, \bfy$ and $\bfz$ holds that $I\bigl(\bfx : \bfy \bigr) \leq I\bigl(\bfx : \bfy\, \bfz\bigr)$.
\end{corollary}

\noindent The $\ell_1$-distance between two distributions is defined as
\[
\bigl\| \mu'(\rmz) - \mu(\rmz) \bigr\|_1 = \sum_{z} \bigl| \mu'(z) - \mu(z)\bigr|.
\]
There is a relation between $\ell_1$-distance and Kullback-Leibler divergence.
\begin{fact}[Pinsker's Inequality]
 For any two distributions $\mu'(\rmz)$ and $\mu(\rmz)$, we have
 \[
   \bigl\| \mu'(\rmz) - \mu(\rmz) \bigr\|_1 \leq \sqrt{2\cdot \DKL \bigl( \mu'(\rmz) ~\|~ \mu(\rmz)\bigr)}
 \]
\end{fact}

\section{The Proof of Theorem~\ref{thm:MainResult}}\label{sec:proof}
Let $f : \{0,1\}^n \times \{0,1\}^n \to \{0,1, *\}$ be a partial function.
Suppose there is a semi-adaptive random scheme $D$ for the multiphase problem of $f$ with error probability bounded by $\tilde{\varepsilon}$ such that $\tu \cdot n \leq o\bigl(k/w\bigr)$.
Let $\mu(\rmx, \rmy) = \mu_1(\rmx) \times \mu_2(\rmy)$ be a product distribution over $\{0,1\}^n \times \{0,1\}^n$, such that $\mu(\rmx, \rmy)$ is $\tilde{\alpha}$-balanced according to $f$.
Let $b \in \{0,1\}$ and $g: \{0,1\}^n \times \{0,1\}^n \to \{0,1, *\}$ be a partial function which is $\lambda$-close to $f$ under $\mu$.
We will prove there is a large almost $b$-monochromatic rectangle for $g$.

Let $\bfx_1,\dots,\bfx_k$ be independent random variables each of them distributed according to $\mu_1$ and $\bfy$ be an independent random variable distributed according to $\mu_2$.
Let the random variable $\bfz \in \{0,1\}^m$ and the index $i \in [k]$ be given by Theorem~\ref{thm:rand_var} applied to the random variables $\bfx_1,\dots,\bfx_k,\bfy$ and the function $f$.
For simplicity we denote $\bfx = \bfx_i$.

We will denote the joint distribution of $(\bfx_1, \ldots, \bfx_k, \bfy, \bfz)$ by $\mu(\rmx_1, \ldots, \rmx_k, \rmy, \rmz)$. Note that here the notation is consistent, in the sense that $\mu(x_i, y) = \mu_1(x_i) \times \mu_2(y)$ for all $i \in [k], x, y \in \{0,1\}^n$. We will then need to keep in mind that $\mu(\rmz)$ is the $\rmz$-marginal of the joint distribution of $(\bfx_1, \ldots, \bfx_k, \bfy, \bfz)$.

By $f(\bfx, \bfy) \neq^* \bfz_m$ we denote the event that the random variable $\bfz_m$ gives us the wrong answer on an input from the support of $f$, i.e. $f(\bfx, \bfy) \neq *$ and $f(\bfx, \bfy) \neq \bfz_m$ hold simultaneously.
By Theorem~\ref{thm:rand_var} we know that $\Pr\bigl[f(\bfx, \bfy) \neq^* \bfz_m\bigr] \leq \tilde{\varepsilon}.$
Since $f$ and $g$ are $\lambda$-close under $\mu$, we have that $\mu$ is still balanced according to $g$ and $g(\bfx, \bfy) \neq^* \bfz_m$ with small probability, as stated in the next observation.

\begin{observation}
\label{obs:CloseFunction}
  Let $\alpha = \tilde{\alpha} - \lambda$ and $\varepsilon = \tilde{\varepsilon} + \lambda$.
 For the function $g$ it holds that
 \begin{enumerate}
  \item The distribution $\mu(\rmx,\rmy)$ is $\alpha$-balanced according to $g$.
  \item $\Pr\bigl[g(\bfx, \bfy) \neq^* \bfz_m\bigr] \leq \varepsilon$.
 \end{enumerate}
\end{observation}
\begin{proof}
 Let $b' \in \{0,1\}$.
 We will bound $\mu\bigl(g^{-1}(b')\bigr)$.
 \begin{align*}
  \tilde{\alpha} \leq \Pr\bigl[f(\bfx, \bfy) = b'\bigr] = &\Pr\bigl[f(\bfx, \bfy) = b', f(\bfx, \bfy) = g(\bfx, \bfy)\bigr] \\
  &+ \Pr\bigl[f(\bfx, \bfy) = b', f(\bfx, \bfy) \neq g(\bfx, \bfy)\bigr] \\
  \leq &\Pr\bigl[g(\bfx, \bfy) = b'\bigr] + \lambda.
 \end{align*}
 Thus, by rearranging we get $\mu\bigl(g^{-1}(b')\bigr) \geq \tilde{\alpha} - \lambda = \alpha$.
 The proof of the second bound is similar:
 \begin{align*}
  \Pr\bigl[g(\bfx, \bfy) \neq^* \bfz_m\bigr] = &\Pr\bigl[f(\bfx, \bfy) \neq^* \bfz_m, f(\bfx, \bfy) = g(\bfx, \bfy)\bigr] \\
  &+ \Pr\bigl[g(\bfx, \bfy) \neq^* \bfz_m, f(\bfx, \bfy) \neq g(\bfx, \bfy)\bigr]   \leq \tilde{\varepsilon} + \lambda = \varepsilon. \qedhere
 \end{align*}
\end{proof}

Let $c$ be the bound on $I\bigl(\bfx : \bfy\,\bfz \bigr)$ and $I\bigl(\bf\bfy : \bfz\bigr)$ given by Theorem~\ref{thm:rand_var}.
Since $I\bigl(\bfx : \bfz\bigr) \leq I\bigl(\bfx : \bfy\, \bfz\bigr)$, we have
$ I\bigl(\bfx : \bfz\bigr), I\bigl(\bfy : \bfz\bigr) \leq \tq\cdot w + o(\tq\cdot w) = c.$
We will prove that if we assume that $\tu \cdot n < o\bigl(k/w\bigr)$ and we choose $p$ large enough ($p$ of Theorem \ref{thm:rand_var}) then we can find a rectangle $R \subseteq X\times Y$ such that $R$ is $O\bigl(\varepsilon/\alpha\bigr)$-error $b$-monochromatic for $g$ and $\mu(R) \geq \frac{1}{2^{c'}}$ for $c' = O\bigl(\frac{\tq \cdot w}{\alpha}\bigr)$.
Thus, we have $\mono^{O(\varepsilon/\alpha)}_\mu(g) \geq 2^{-c'}$ and consequently 
\[
 \scb^{O(\varepsilon/\alpha),\lambda}_\mu(f) \leq O\left(\frac{\tq\cdot w}{\alpha}\right).
\]
By rearranging, we get the bound of Theorem~\ref{thm:MainResult}.

Let us sketch the proof of how we can find such a rectangle $R$.
We will first fix the random variable $\bfz$ to $z$ such that $\bfx$ and $\bfy$ are not very correlated conditioned on $\bfz = z$, i.e., the joint distribution $\mu(\rmx, \rmy \mid z)$ is very similar to the product distribution of the marginals $\mu(\rmx \mid z) \times \mu(\rmy \mid z)$.
Moreover, we will pick $z$ in such a way the probability of error $\Pr\bigl[g(\bfx, \bfy) \neq^* z_m | \bfz = z\bigr]$ is still small.
Then, since $\mu(\rmx, \rmy \mid z)$ is close to $\mu(\rmx \mid z) \times \mu(\rmy \mid z)$, the probability of error under the latter distribution will be small as well, i.e., if $(\bfx', \bfy') \sim \mu(\rmx \mid z) \times \mu(\rmy \mid z)$, then $\Pr\bigl[g(\bfx', \bfy') \neq^* z_m \bigr]$ will also be small.
Finally, we will find subsets $A \subseteq \supp\bigl(\mu(\rmx \mid z)\bigr), B \subseteq \supp\bigl(\mu(\rmy \mid z)\bigr)$ of large mass (under the original distributions $\mu_1$ and $\mu_2$), while keeping the probability of error on the rectangle $R = A \times B$ sufficiently small.

Let us then proceed to implement this plan. Let $\beta = \alpha - \varepsilon$. We will show that $\beta$ is a lower bound for the probability that $\bfz_m$ is equal to $b$. Let $\gamma$ be the bound on $I\bigl(\bfx : \bfy \mid \bfz\bigr)$ given by Theorem~\ref{thm:rand_var}, i.e., $I\bigl(\bfx_i : \bfy \mid \bfz\bigr) \leq \gamma = O\left(\frac{\tu \cdot n \cdot w}{p}\right)$.

\begin{lemma}
\label{lem:GoodSet}
 There exists $z \in Z$ such that
 \begin{enumerate}
  \item $z_m = b$.
  \item $I\bigl(\bfx : \bfy \mid \bfz = z\bigr) \leq \frac{5}{\beta}\cdot\gamma$.
  \item $\DKL\bigl(\mu(\rmx \mid z)  ~\|~ \mu(\rmx)\bigr), \DKL\bigl(\mu(\rmy \mid z) ~\|~ \mu(\rmy)\bigr) \leq \frac{5}{\beta}\cdot c$.
  \item $\Pr\bigl[g(\bfx, \bfy) \neq^* z_m \mid \bfz = z\bigr] \leq \frac{5}{\beta}\cdot \varepsilon$.
 \end{enumerate}
\end{lemma}
% \begin{proof}[Proof sketch]
%   Since $\mu$ is $\alpha$-balanced according to $g$, we find that
%  \begin{align*}
%   \alpha &\leq \Pr\bigl[g(\bfx, \bfy)= b\bigr] \\
%   &= \Pr\bigl[g(\bfx, \bfy) = b, \bfz_m = b \bigr]   + \Pr\bigl[g(\bfx, \bfy) = b, \bfz_m \neq b\bigr] \leq \Pr\bigl[\bfz_m = b \bigr] + \varepsilon.
%  \end{align*}
%  Thus, by rearranging we get $\Pr\bigl[\bfz_m = b \bigr] \geq \alpha - \varepsilon = \beta$.
 
%  The rest of inequalities are true with high probability due to the Markov inequality.
%  Finally by union bound, we get all the inequalities and that $\bfz_m = b$ hold with non-zero probability. We ommit the rest of the proof.
%  % The full proof is in Appendix.
% \end{proof}

\begin{proof}
 Since $\mu$ is $\alpha$-balanced according to $g$, we find that
 \begin{align*}
  \alpha &\leq \Pr\bigl[g(\bfx, \bfy)= b\bigr] \\
  &= \Pr\bigl[g(\bfx, \bfy) = b, \bfz_m = b \bigr]   + \Pr\bigl[g(\bfx, \bfy) = b, \bfz_m \neq b\bigr] \leq \Pr\bigl[\bfz_m = b \bigr] + \varepsilon.
 \end{align*}
 Thus, by rearranging we get $\Pr\bigl[\bfz_m = b \bigr] \geq \alpha - \varepsilon = \beta$.
 By expanding the information $I\bigl(\bfx : \bfy \mid \bfz\bigr)$ we find
 \begin{align*}
   &\gamma \geq I\bigl(\bfx : \bfy \mid \bfz\bigr) = \expec_{z \sim \mu(\rmz)} \Big[ I\bigl(\bfx : \bfy \mid \bfz = z\bigr) \Big]\\
   \intertext{and by the Markov inequality we get that}
  &\Pr_{z \sim \mu(\rmz)}\left[I\bigl(\bfx : \bfy \mid \bfz = z\bigr) \geq \frac{5}{\beta}\cdot \gamma\right] \leq \frac{\beta}{5}.
 \end{align*}
 Similarly, for the information $I\bigl(\bfx : \bfz\bigr)$:
 \begin{align*}
   & c \geq I\bigl(\bfx\,\bfy : \bfz\bigr) \geq I\bigl(\bfx : \bfz\bigr) = \expec_{z \sim \mu(\rmz)} \Big[  \DKL\bigl(\mu(\rmx \mid z)  ~\|~ \mu(\rmx)\bigr) \Big] \\
   \intertext{and so}
  &\Pr_{z \sim \mu(\rmz)}\left[\DKL\bigl(\mu(\rmx \mid z)  ~\|~ \mu(\rmx)\bigr) \geq \frac{5}{\beta}\cdot c\right] \leq \frac{\beta}{5}.
 \end{align*}
 The bound for $I\bigl(\bfy : \bfz\bigr)$ is analogous.
 Let $e_z = \Pr_\mu\bigl[g(\bfx, \bfy) \neq^* z_m | \bfz = z\bigr]$.
 Then,
 \begin{align*}
  & \varepsilon \geq \Pr\bigl[g(\bfx, \bfy) \neq^* {\bf z}_m\bigr] = \sum_{z \in Z} \mu(z) \cdot e_z = \expec_{z \sim \mu(\rmz)} \big[ e_z \big] \\
  & \Pr_{z \sim \mu(\rmz)} \left[ e_z \geq \frac{5}{\beta}\cdot \varepsilon \right] \leq \frac{\beta}{5}.
 \end{align*}
 Thus, by a union bound we may infer the existence of the sought $z \in Z$.
\end{proof}

Let us now fix $z \in Z$ from the previous lemma. Let $\mu_z(\rm x, \rmy) = \mu(\rmx, \rmy \mid z)$ be the distribution $\mu(\rmx, \rmy)$ conditioned on $\bfz = z$, and let $\mu_z'(\rmx, \rmy) = \mu(\rmx \mid z) \times \mu(\rm y \mid z)$ be the product of its marginals. 
Let $S$ be the support of $\mu_z(\rmx, \rmy)$, and let $S_\rmx$ and $S_\rmy$ be the supports of $\mu_z'(\rmx)$ and $\mu_z'(\rmy)$, respectively, i.e., $S_\rmx$ and $S_\rmy$ are the projections of $S$ into $X$ and $Y$.

Then Pinsker's inequality will give us that $\mu_z$ and $\mu_z'$ are very close.
Let $\delta = \sqrt{\frac{10}{\beta}\cdot\gamma}$.

\begin{lemma}
\label{lem:TotalVar}
 $\Bigl\|\mu_z(\rmx, \rmy) - \mu_z'(\rmx, \rmy)\Bigr\|_1 \leq \delta$
\end{lemma}
% \begin{proof}[Proof sketch]
%  The proof is by Lemma~\ref{lem:GoodSet} and by standard application of Pinsker's inequality and the definition of mutual information. We ommit it here.
% % See Appendix for the full proof.
% \end{proof}

\begin{proof}
  Indeed, by Pinsker's inequality,
  \[
    \Bigl\|\mu_z(\rmx, \rmy) - \mu_z'(\rmx, \rmy)\Bigr\|_1 \le \sqrt{2\cdot \DKL \bigl( \mu_z(\rmx, \rmy) ~\|~ \mu_z'(\rmx, \rmy)  \bigr)}.
  \]
  The right-hand side is $\sqrt{2\cdot \DKL \bigl( \mu(\rmx, \rmy \mid z) ~\|~ \mu(\rmx \mid z) \times \mu(\rm y \mid z)  \bigr)}$, which by definition of mutual information equals $\sqrt{2\cdot I\bigl(\bfx : \bfy \mid \bfz = z\bigr)}$, and by Lemma~\ref{lem:GoodSet} this is $\le \sqrt{\frac{10}{\beta}\cdot\gamma} = \delta$.
\end{proof}

% As a corollary of Lemma~\ref{lem:TotalVar} we will prove that $g(x,y)$ differs from $z_m$ with small probability not only under the distribution $\mu_z(\rmx, \rmy) = \mu(\rmx, \rmy \mid z)$ but also under the distribution $\mu_z'(\rmx, \rmy) = \mu(\rmx \mid z) \times \mu(\rm y \mid z)$.

For the sake of reasoning, let $(\bfx', \bfy') \sim \mu_z'(\rmx, \rmy)$ be random variables chosen according to to $\mu_z'$. Let $\varepsilon' = \frac{5}{\beta}\cdot \varepsilon + \delta$.
It then follows from Lemma~\ref{lem:GoodSet} and Lemma~\ref{lem:TotalVar} that:

\begin{lemma}
\label{lem:SmallError}
$\Pr\bigl[g(\bfx', \bfy') \neq^* z_m\bigr] \leq \varepsilon'$. 
\end{lemma}
\begin{proof}
 We prove that 
 \[
    \Bigl|\Pr\bigl[g(\bfx, \bfy) \neq^* z_m \mid \bfz = z \bigr] - \Pr\bigl[g(\bfx', \bfy') \neq^* z_m\bigr] \Bigr| \leq \delta.  
 \]
 Since $\Pr\bigl[g(\bfx,\bfy) \neq^* z_m \mid \bfz = z\bigr] \leq \frac{5}{\beta} \cdot \varepsilon$ by Lemma~\ref{lem:GoodSet}, the lemma follows.
 Let 
 \[
B = \bigl\{ (x,y) \in S_\rmx \times S_\rmy : g(x,y) \neq z_m, g(x,y) \neq * \bigr\}.  
 \]
 \begin{align*}
  \Bigl|\Pr&\bigl[g(\bfx, \bfy) \neq^* z_m \mid \bfz = z \bigr] - \Pr\bigl[g(\bfx', \bfy') \neq^* z_m\bigr] \Bigr| &\\
  &= \Bigl| \sum_{(x,y) \in B} \mu_z(x,y) - \mu'_z(x,y) \Bigr| &\\
  &\leq  \sum_{(x,y) \in B} \Bigl| \mu_z(x,y) - \mu'_z(x,y) \Bigr| \leq \delta\tag*{\text{by the triangle inequality and Lemma~\ref{lem:TotalVar}}}
 \end{align*}
\end{proof}

Let $c' = \frac{5}{\beta}\cdot c$.
We will prove the ratio between $\mu_z'(\bfx')$ and $\mu(\bfx')$ is larger than $2^{O(c')}$ with only small probability (when $\bfx' \sim \mu_z'(\rmx)$). 
The same holds for $\mu_z'(\bfy')$ and $\mu(\bfy')$.

\begin{lemma}
\label{lem:SmallProbability} 
$\Pr\left[\mu_z'(\bfx') \geq 2^{6c'}\cdot \mu(\bfx') \right], \Pr \left[\mu_z'(\bfy') \geq 2^{6c'}\cdot \mu(\bfy') \right] \leq \frac{1}{6}.$
\end{lemma}
\begin{proof}
 We prove the lemma for $\mu_z'(\bfx')$, the proof for $\mu_z'(\bfy')$ is analogous. 
 By Lemma~\ref{lem:GoodSet} we know that $\DKL\bigl( \mu(\rmx\mid z)  ~\|~ \mu(\rmx)\bigr) = \DKL\bigl( \mu_z(\rmx)  ~\|~ \mu(\rmx)\bigr) = \DKL\bigl( \mu'_z(\rmx)  ~\|~ \mu(\rmx)\bigr) \leq c'$.
 We expand the Kullback-Leibler divergence:
 \begin{align*}
   &c' \geq \DKL\bigl( \mu'_z(\rmx)  ~\|~ \mu(\rmx)\bigr) = \sum_{x \in S_\rmx} \mu_z'(x) \log \frac{\mu_z'(x)}{\mu(x)} = \expec\left[ \log \frac{\mu_z'(\bfx')}{\mu(\bfx')}\right], \\
   \intertext{and then use the Markov inequality:}
  &\Pr \left[\mu_z'(\bfx') \geq 2^{6c'}\cdot \mu(\bfx') \right] = \Pr \left[ \log \frac{\mu_z'(\bfx')}{\mu(\bfx')} \geq 6c'\right] \leq \frac{1}{6}. \qedhere
 \end{align*}
\end{proof}
We now split $S_\rmx$ and $S_\rmy$ into buckets $C^\rmx_\ell$ and $C^\rmy_\ell$ (for $\ell \geq 1$), where the $\ell$-th buckets are
\begin{align*}
 C^\rmx_\ell &= \left\{x \in S_\rmx \;\Big|\; \frac{(\ell - 1)}{2^{c'}} <  \frac{\mu_z'(x)}{\mu(x)} \leq \frac{\ell}{2^{c'}} \right\}, \\
 C^\rmy_\ell &= \left\{y \in S_\rmy \;\Big|\; \frac{(\ell - 1)}{2^{c'}} <  \frac{\mu_z'(y)}{\mu(y)} \leq \frac{\ell}{2^{c'}} \right\}.
\end{align*}
In a bucket $C^\rmx_\ell$ there are elements of $S_\rmx$ such that their probability under $\mu_z'(\rmx)$ is approximately $\frac{\ell}{2^{c'}}$-times bigger than their probability under $\mu(\rmx)$.
By Lemma~\ref{lem:SmallProbability}, it holds that with high probability the elements $x \in S_\rmx, y \in S_\rmy$ are in the buckets $C^\rmx_{\ell_1}$ and $C^\rmy_{\ell_2}$ for $\ell_1,\ell_2 \leq 2^{7c'}$.
Thus, if we find a bucket $C^\rmx_{\ell_1}$ for ${\ell_1} \leq 2^{7c'}$ which has probability at least $\frac{1}{2^{O(c')}}$ under $\mu_z'(\rmx)$, then it has also probability at least $\frac{1}{2^{O(c')}}$ under $\mu(\rmx)$.
The same holds also for buckets $C^\rmy_\ell$.
In the next lemma we will show that there are buckets $C^\rmx_{\ell_1}$ and $C^\rmy_{\ell_2}$ of large probability under $\mu_z'$ such that the probability of error on $C^\rmx_{\ell_1} \times C^\rmy_{\ell_2}$ is still small.

\begin{lemma}
\label{lem:BigBuckets}
 There exist buckets $C^\rmx_{\ell_1}$ and $C^\rmy_{\ell_2}$ such that
 \begin{enumerate}
  \item $1 < \ell_1, \ell_2 \leq 2^{7c'}$.
  \item $\Pr\bigl[\bfx' \in C^\rmx_{\ell_1}\bigr], \Pr\bigl[\bfy' \in C^\rmy_{\ell_2}\bigr] \geq \frac{1}{6\cdot 2^{7c'}}$.
  \item $\Pr\bigl[g(\bfx',\bfy') \neq^* z_m, (\bfx', \bfy') \in C^\rmx_{\ell_1} \times C^\rmy_{\ell_2}\bigr] \leq 6\varepsilon' \cdot \Pr\bigl[(\bfx', \bfy') \in C^\rmx_{\ell_1} \times C^\rmy_{\ell_2}\bigr].$
 \end{enumerate}
\end{lemma}
\begin{proof}
 We prove that $\ell_1, \ell_2$ exist via the probabilistic method. Let ${\bfell}_1$ and ${\bfell}_2$ be the buckets of $\bfx'$ and $\bfy'$, respectively. Thus
 $ \Pr\bigl[{\bfell}_1 = \ell\bigr] = \Pr\bigl[\bfx' \in C^\rmx_\ell\bigr]$ and $ \Pr\bigl[{\bfell}_2 = \ell\bigr] = \Pr\bigl[\bfy' \in C^\rmy_\ell\bigr].$
 
 Let $B_1, B_2 \subseteq L' = \{1,\dots, 2^{7c'}\}$ be sets of indices of small probability, i.e., for $i \in \{1,2\}$
 \begin{align*}
  B_i = \left\{\ell \in L' ~\bigl|~ \Pr[\bfell_i = \ell] \leq \frac{1}{6\cdot 2^{7c'}} \right\}.
 \end{align*}
 We will prove that with high probability we have $2^{7c'} \geq {\bfell}_1 > 1$ and ${\bfell}_1 \not \in B_1$.
 The proof for ${\bfell}_2$ is analogous.
 \[
  \Pr\bigl[{\bfell}_1 = 1\bigr] = \Pr\bigl[\bfx' \in C^\rmx_1\bigr] = \sum_{x \in C^\rmx_1} \mu_z'(x) \leq \frac{\sum_{x \in C^\rmx_1} \mu(x)}{2^{c'}} \leq \frac{1}{2^{c'}}
 \]
 By Lemma~\ref{lem:SmallProbability}, we get $\Pr\bigl[{\bfell}_1 > 2^{7c'}\bigr] = \Pr\bigl[\mu_z'(\bfx') \geq 2^{6c'}\cdot \mu(\bfx')\bigr] \leq \frac{1}{6}.$
 There is only small probability that ${\bfell}_1$ is in $B_1$.
 \[
  \Pr\bigl[{\bfell}_1 \in B_1\bigr] = \sum_{\ell \in B_1} \Pr[{\bfell}_1 = \ell] \leq \frac{|L'|}{6 \cdot 2^{7c'}} = \frac{1}{6}
 \]
 Thus, we have that ${\bfell}_i \in B_i$ or ${\bfell}_i = 1$ or ${\bfell}_i > 2^{7c'}$ with probability at most $\frac{2}{3} + \frac{2}{2^{c'}}$. 
 
 By Lemma~\ref{lem:SmallError}, we have that $\Pr\bigl[g(\bfx', \bfy') \neq^* z_m\bigr] \leq \varepsilon'$.
 By expanding the probability and by Markov inequality we will now get the last inequality for $C^\rmx_{\ell_1}$ and $C^\rmy_{\ell_2}$.
 Let 
 \[
    e(\ell_1, \ell_2) = \Pr\bigl[g(\bfx', \bfy') \neq^* z_m \mid \bfx' \in C^\rmx_{\ell_1}, \bfy' \in C^\rmy_{\ell_2}\bigr].
 \]
 We will prove there is $\ell_1$ and $\ell_2$ such that $e(\ell_1, \ell_2) \leq 6\varepsilon'$.
 This is equivalent to the third bound of the lemma. We have:
  $\varepsilon' \geq \Pr\bigl[g(\bfx', \bfy') \neq^* z_m\bigr] =  \expec \big[ e(\bfell_1, \bfell_2) \big]$
   and thus, by Markov,
  $\Pr\bigl[e({\bfell}_1,{\bfell}_2) > 6\varepsilon'\bigr] \leq \frac{1}{6}.$
 By a union bound we conclude that there must exist $1 < \ell_1, \ell_2 \leq 2^{7c'}$ such that $\Pr[\bfell_1 = \ell_1], \Pr[\bfell_2 = \ell_2] \geq \frac{1}{6\cdot 2^{7c'}}$ and $e(\ell_1,\ell_2) \leq 6\varepsilon'$. 
\end{proof}

As a corollary we will prove that the rectangle $C^\rmx_{\ell_1} \times C^\rmy_{\ell_2}$ (given by the previous lemma) is a good rectangle under the original distribution $\mu$. We remark that the proof of the following corollary is the only place where we use the fact that $\bfx$ and $\bfy$ are independent.

\begin{corollary}
\label{cor:large_rectangle}
 There exists a rectangle $R \subseteq S_\rmx \times S_\rmy$ such that 
 \begin{enumerate}
  \item $\Pr\bigl[(\bfx,\bfy) \in R\bigr] \geq \frac{1}{36\cdot 2^{26c'}}.$
  \item $\Pr\bigl[g(\bfx, \bfy) \neq^* z_m, (\bfx, \bfy) \in R\bigr] \leq 24\varepsilon' \cdot \Pr\bigl[(\bfx,\bfy) \in R\bigr]$.
 \end{enumerate}
\end{corollary}
\begin{proof}
 Let $R = C^\rmx_{\ell_1} \times C^\rmy_{\ell_2}$ where $C^\rmx_{\ell_1}$ and $C^\rmy_{\ell_2}$ are buckets given by Lemma~\ref{lem:BigBuckets}.
 By Lemma~\ref{lem:BigBuckets}, we get
 \begin{align*}
  \frac{1}{6\cdot 2^{7c'}} \leq \Pr\bigl[\bfx' \in C^\rmx_{\ell_1}\bigr] = \sum_{x \in C^\rmx_{\ell_1}} \mu_z'(x) \leq \sum_{x \in C^\rmx_{\ell_1}} \frac{\ell_1 \cdot \mu(x)}{2^{c'}} = \Pr\bigl[\bfx \in C^\rmx_{\ell_1}\bigr]\cdot \frac{\ell_1}{2^{c'}}.
 \end{align*}
 By rearranging we get
 \[
 \Pr\bigl[\bfx \in C^\rmx_{\ell_1}\bigr] \geq \frac{2^{c'}}{6\ell_1\cdot 2^{7c'}} \geq \frac{1}{6\cdot 2^{13c'}}
 \]
 The bound for $\Pr\bigl[\bfy \in C^\rmy_{\ell_2}\bigr]$ is analogous, thus we have $\Pr\bigl[(\bfx, \bfy) \in R\bigr] \geq \frac{1}{36\cdot 2^{26c'}}.$ (Here and below, we crucially use the fact that $\bfx, \bfy$ are given by a product distribution.)
 Now we prove the second bound for $R$.
 Let $B = \bigl\{ (x,y) \in R : g(x,y) \neq z_m, g(x,y) \neq *\bigr\}.$
 \begin{align*}
  6\varepsilon' \cdot \Pr\bigl[(\bfx, \bfy) \in R\bigr] \cdot \frac{\ell_1 \ell_2}{2^{2c'}} &\geq 6\varepsilon' \cdot \Pr\bigl[(\bfx', \bfy') \in R\bigr] & \text{by definition of buckets} \\
  &\geq  \Pr\bigl[(\bfx', \bfy') \in B\bigr]& \text{by Lemma~\ref{lem:BigBuckets}} \\
  &\geq \Pr\bigl[(\bfx, \bfy) \in B\bigr] \cdot \frac{(\ell_1 - 1)(\ell_2 - 1)}{2^{2c'}}& \text{by definition of buckets}
 \end{align*}
 Thus, by rearranging we get
 \[
  \Pr\bigl[(\bfx, \bfy) \in B] \leq 6\varepsilon'\cdot \Pr\bigl[(\bfx, \bfy) \in R\bigr] \cdot \frac{\ell_1\ell_2}{(\ell_1 - 1)(\ell_2 - 1)} \leq 24\varepsilon'\cdot \Pr\bigl[(\bfx, \bfy) \in R\bigr],
 \]
 as $\frac{\ell_1\ell_2}{(\ell_1 - 1)(\ell_2 - 1)} \leq 4$ for $\ell_1,\ell_2 > 1$ by Lemma~\ref{lem:BigBuckets}.
\end{proof}

\begin{proof}[Proof of Theorem~\ref{thm:MainResult}]
Suppose that $\tu \cdot n \leq o\bigl(k / w\bigr)$.
Let $R$ be the rectangle given by Corollary~\ref{cor:large_rectangle}.
It holds that the rectangle $R$ is $24\varepsilon'$-error $b$-monochromatic for $g$ under $\mu$.
Therefore, for the function $g$ holds that
\begin{equation}
 \label{eq:Final}
 \mono^{24\varepsilon'}_\mu(g) \geq \Pr\bigl[(\bfx, \bfy) \in R\bigr] \geq \frac{1}{36\cdot 2^{26c'}}.
\end{equation}
We need to argue that $\varepsilon'$ is $O(\varepsilon/\alpha)$.
By definition,
\[
 \varepsilon' = \frac{5}{\alpha - \varepsilon}\cdot\varepsilon + \delta.
\]
We recall that
\[
 \delta = O\left(\sqrt{\frac{\tu\cdot n \cdot w}{p}}\right) \leq  \sqrt{\frac{o(k)}{p}}.
\]
Thus, we can set $p$ to be large enough so that $\delta$ is smaller than arbitrary constant and still $p \leq o(k)$.
By the assumption we have $2\varepsilon < \alpha$.
Thus, $\frac{\varepsilon}{\alpha - \varepsilon} \leq \frac{2\varepsilon}{\alpha}$ and we conclude that $\varepsilon'$ is $O\bigl(\varepsilon/\alpha\bigr)$.
Since $c' = O\bigl(\frac{\tq\cdot w}{\alpha\cdot(1-\varepsilon)}\bigr) = O\bigl(\frac{\tq\cdot w}{\alpha}\bigr)$, we get the result by rearranging Inequality~(\ref{eq:Final}).
\end{proof}

\section{Applications}
\label{sec:App}
In this section we apply Theorem~\ref{thm:MainResult} to derive lower bounds for several explicit functions -- Inner Product ($\IP$), Disjointness ($\Disj$), Gap Orthogonality ($\ORT$) and Gap Hamming Distance ($\GH$):

\begin{align*}
 \IP(x,y) &= \sum_{i \in n} x_i \cdot y_i \mod 2, \\
 \GH_n(x,y) &=
 \begin{cases}
  1 & \text{if $\Delta_H(x,y) \geq \frac{n}{2} + \sqrt{n} $}, \\
  0 & \text{if $\Delta_H(x,y) \leq \frac{n}{2}  - \sqrt{n}$}. \\
 \end{cases}
\end{align*}

The function $\Delta_H$ is the Hamming Distance of two strings, i.e., $\Delta_H(x,y)$ is a number of indices $i \in [n]$ such that $x_i \neq y_i$.
For $\IP_\R(x,y) = \sum_{i \in [n]} (-1)^{x_i + y_i}$ we define
\[
 \ORT_{n,d} (x,y) =
 \begin{cases}
  1 & \text{ if } \bigl|\IP_\R(x,y) \bigr| \geq 2d\cdot \sqrt{n} \\
  0 & \text{ if } \bigl|\IP_\R(x,y) \bigr| \leq d\cdot\sqrt{n}.
 \end{cases}
\]
The standard value for $d$ is 1, thus we denote $\ORT_n = \ORT_{n,1}$.
Note that $\Delta_H(x,y) = \frac{n - \IP_\R(x,y)}{2}$ and $\IP_\R(x,y)$ is the Inner Product of $x', y'$ over $\R$ where $x'$ and $y'$ arise from $x$ and $y$ by replacing $0$ by $1$ and $1$ by $-1$.
We present previous results with bounds for measures of interest under hard distributions.

\begin{theorem}[\cite{kushilevitz96}]
 Let $\mu_1$ be a uniform distribution on $\{0,1\}^n \times \{0,1\}^n$. 
 Then,
 \[
  \disc_{\mu_1}(\IP) \leq \frac{1}{2^{n/2}}.
 \]
\end{theorem}

\begin{theorem}[Babai et al.~\cite{babai86}]
 Let $\rho < 1/100$ and $\mu_2$ be a a uniform distribution over $S \times S$, where $S$ consists of $n$-bit strings containing exactly $\sqrt{n}$ 1's.
 Then,
 \[
  \mono^\rho_{\mu_2}(\Disj) \leq \frac{1}{2^{\Omega(\sqrt{n})}}.
 \]
\end{theorem}
Sherstov~\cite{sherstov11} provided a lower bound of communication complexity of $\GH$ by lower bound of corruption bound of $\ORT_{n,\frac{1}{8}}$ following by reduction to $\GH$.
\begin{theorem}[Sherstov~\cite{sherstov11}]
 Let $\rho > 0$ be sufficiently small and $\mu_3$ be a uniform distribution over $\{0,1\}^n \times \{0,1\}^n$.
 Then,
 \[
  \cb^\rho_{\mu_3}(\ORT_{n,\frac{1}{8}}) \geq \rho \cdot n.
 \]
\end{theorem}

By this theorem and Theorem~\ref{thm:MainResult} we get a lower bound for data structures for $\ORT_{n,\frac{1}{8}}$.
By reductions used by Sherstov~\cite{sherstov11} we also get a lower bounds for $\ORT$ and $\GH$.
\begin{align*}
 \ORT_{n,\frac{1}{8}}(x,y) =~ &\ORT_{64n}\bigl(x^{64},y^{64}\bigr) \\
 \ORT_n(x,y) = ~&\GH_{10n + 15\sqrt{n}}\bigl(x^{10}1^{15\sqrt{n}}, y^{10}0^{15\sqrt{n}}\bigr) \\
 &\wedge \neg\GH_{10n + 15\sqrt{n}}\bigl(x^{10}0^{15\sqrt{n}}, y^{10}0^{15\sqrt{n}}\bigr)
\end{align*}
Where $s^i$ denote $i$ copies of $s$ concatenated together.
Let $D$ be a semi-adaptive random scheme for the multiphase problem of the presented functions with sufficiently small error probability.
By the theorems presented in this section and by Theorem~\ref{thm:MainResult}, we can derive the following lower bounds for $\tq\cdot w$, assuming that $\tu \cdot n \leq o\bigl(k / w\bigr)$.

% \begin{center}
% \begin{tabular}{|c|c|c|}
% \hline
%  Function $f$         & Lower bound  \\
%                & for $\tq\cdot w$ \\
%  \hline
%  $\IP$             & $\Omega(n)$ \\
%  $\Disj$          & $\Omega(\sqrt{n})$ \\
%  $\ORT_n$                               & $\Omega(n)$ \\
%  $\GH_n$                               & $\Omega(n)$ \\
%  \hline
% \end{tabular}
% \end{center}

\begin{center}
\begin{tabular}{|c|c|c|}
\hline
 Function $f$ & Ballancedness            & Lower bound  \\
              & of the hard distribution & of $\tq\cdot w$ \\
 \hline
 $\IP$      & $\frac{1}{2}$            & $\Omega(n)$ \\
 $\Disj$    & $\sim \frac{1}{e}$       & $\Omega(\sqrt{n})$ \\
 $\ORT_n$     &    $\Theta(1)$                      & $\Omega(n)$ \\
 $\GH_n$      &         N/A (lower-bound is via reduction)             & $\Omega(n)$ \\
 \hline
\end{tabular}
\end{center}

\bibliography{multiphase}

%\clearpage
%\appendix
%\input{appendix}

\end{document}